\documentclass[a4paper,12pt]{article}

\usepackage[utf8]{inputenc}
\usepackage[T1]{fontenc}
\usepackage[a4paper, margin=2.5cm]{geometry}

\usepackage{amsmath, amsthm, amssymb, enumerate}
\usepackage{graphicx}
\usepackage{enumerate}
\usepackage{authblk}
\usepackage[textwidth=2.1cm, color=orange!70]{todonotes}
\usepackage[absolute]{textpos}
\usepackage{comment}
\usepackage{dsfont}
\usepackage{needspace}
\usepackage{mathtools}

\usepackage{thmtools}

\usepackage{tikz}
\usetikzlibrary{decorations.pathreplacing}
\tikzstyle{small node} = [draw, circle, fill = black, minimum size = 3pt, inner sep = 0pt]
\tikzstyle{black node} = [draw, circle, fill = black, minimum size = 5pt, inner sep = 0pt]
\tikzstyle{white node} = [draw, circle, fill = white, minimum size = 5pt, inner sep = 0pt]
\tikzstyle{normal} = [draw=none, fill = none, rectangle, minimum size =0]

\usepackage{ifthen}
\def\empty{}
\pgfkeys{utils/.cd,
	set if not empty/.code n args={3}{%
		\def\arg{#2}%
		\ifx\arg\empty%
		\pgfkeys{#1={#3}}%
		\else%
		\pgfkeys{#1={#2}}%
		\fi%
	},
	set if labelpos not empty/.code n args={2}{%
		\def\arg{#1}%
		\ifx\arg\empty%
		\else%
		\def\argo{#2}%
		\ifx\argo\empty%
		\pgfkeys{/tikz/label={#1}}%
		\else%
		\pgfkeys{/tikz/label={#2}:{#1}}%
		\fi%
		\fi%
	}
}
\tikzset{
	nodes/.style n args={4}{
		draw ,circle,outer sep=0.5mm,
		/utils/set if not empty={/tikz/fill}{#1}{black},
		/utils/set if not empty={/tikz/minimum size}{#4}{5}
	}
}

\usepackage{algorithm}
\usepackage[noend]{algpseudocode}

\usepackage{hyperref}

\hypersetup{
  pdftitle = {Avoidable paths in graphs},
  pdfauthor=  {Marthe Bonamy, Oscar Defrain, Meike Hatzel, Jocelyn Thebaut},
  colorlinks = true
}
\usepackage{cleveref}

\newtheorem{environment}{Environment}[section]

\newtheorem{lemma}[environment]{Lemma}
\crefname{lemma}{lemma}{lemmas}
\crefformat{lemma}{#2Lemma~#1#3}
\Crefformat{lemma}{#2Lemma~#1#3}

\newtheorem{question}[environment]{Question}
\crefname{question}{question}{questions}
\crefformat{question}{#2Question~#1#3}
\Crefformat{question}{#2Question~#1#3}

\newtheorem{corollary}[environment]{Corollary}
\crefname{corollary}{corollary}{corollaries}
\crefformat{corollary}{#2Corollary~#1#3}
\Crefformat{corollary}{#2Corollary~#1#3}

\newtheorem{theorem}[environment]{Theorem}
\crefname{theorem}{theorem}{theorems}
\crefformat{theorem}{#2Theorem~#1#3}
\Crefformat{theorem}{#2Theorem~#1#3}

\crefname{proposition}{proposition}{Propositions}
\crefformat{proposition}{#2Proposition~#1#3}
\Crefformat{proposition}{#2Proposition~#1#3}

\crefname{conjecture}{conjecture}{Conjectures}
\crefformat{conjecture}{#2Conjecture~#1#3}
\Crefformat{conjecture}{#2Conjecture~#1#3}

\crefname{example}{example}{examples}
\crefformat{example}{#2Example~#1#3}
\Crefformat{example}{#2Example~#1#3}

\crefname{remark}{remark}{remarks}
\crefformat{remark}{#2Remark~#1#3}
\Crefformat{remark}{#2Remark~#1#3}

\crefname{definition}{definition}{definitions}
\crefformat{definition}{#2Definition~#1#3}
\Crefformat{definition}{#2Definition~#1#3}

\crefname{figure}{figure}{figures}
\crefformat{figure}{#2Figure~#1#3}
\Crefformat{figure}{#2Figure~#1#3}

\crefname{chapter}{chapter}{chapters}
\crefformat{chapter}{#2Chapter~#1#3}
\Crefformat{chapter}{#2Chapter~#1#3}

\crefname{section}{section}{sections}
\crefformat{section}{#2Section~#1#3}
\Crefformat{section}{#2Section~#1#3}

\crefname{algorithm}{algorithm}{algorithms}
\crefformat{algorithm}{#2Algorithm~#1#3}
\Crefformat{algorithm}{#2Algorithm~#1#3}

\crefname{notation}{notation}{notations}
\crefformat{notation}{#2Notation~#1#3}
\Crefformat{notation}{#2Notation~#1#3}

\newtheorem{claim}[environment]{Claim}
\crefname{claim}{claim}{claims}
\crefformat{claim}{#2Claim~#1#3}
\Crefformat{claim}{#2Claim~#1#3}

\crefname{enumi}{condition}{conditions}
\crefformat{enumi}{#2Condition~#1#3}
\Crefformat{enumi}{#2Condition~#1#3}

\newtheorem{definitionx}[environment]{Definition}
\newenvironment{definition}
  {\pushQED{\qed}\definitionx}
  {\popQED\enddefinitionx}
\crefname{definitionx}{definition}{definitions}
\crefformat{definitionx}{#2Definition~#1#3}
\Crefformat{definitionx}{#2Definition~#1#3}

\newcommand{\N}{\mathds{N}}

\renewcommand{\geq}{\geqslant}
\renewcommand{\leq}{\leqslant}

\def\cqedsymbol{\ifmmode$\lrcorner$\else{\unskip\nobreak\hfil
\penalty50\hskip1em\null\nobreak\hfil$\lrcorner$
\parfillskip=0pt\finalhyphendemerits=0\endgraf}\fi}

\def\lqedsymbol{\ifmmode$\lrcorner$\else{\unskip\nobreak\hfil
\penalty50\hskip1em\null\nobreak\hfil$\rule{1ex}{1ex}$
\parfillskip=0pt\finalhyphendemerits=0\endgraf}\fi} 

\newcommand{\cqed}{\renewcommand{\qed}{\cqedsymbol}}
\newcommand{\lqed}{\renewcommand{\qed}{\lqedsymbol}}

\title{Avoidable paths in graphs\thanks{The first author has been supported by the ANR project GrR ANR-18-CE40-0032. The second author has been supported by the ANR project GraphEn ANR-15-CE40-0009. The third author was supported by the European Research Council (ERC) under the European Union’s Horizon 2020 research and innovation programme (ERC Consolidator Grant DISTRUCT, grant agreement No 648527).}}

\author[1]{Marthe Bonamy}
\author[2]{Oscar Defrain}
\author[3]{Meike Hatzel}
\author[4]{Jocelyn Thiebaut}
\affil[1]{CNRS, LaBRI, Université de Bordeaux, France.}
\affil[2]{LIMOS, Université Clermont Auvergne, France.}
\affil[3]{LaS, Technische Universität Berlin, Germany.}
\affil[4]{LIRMM, Université de Montpellier, France.}
\date{\today}

\begin{document}

\maketitle

\begin{textblock}{20}(0, 13.1)
\includegraphics[width=40px]{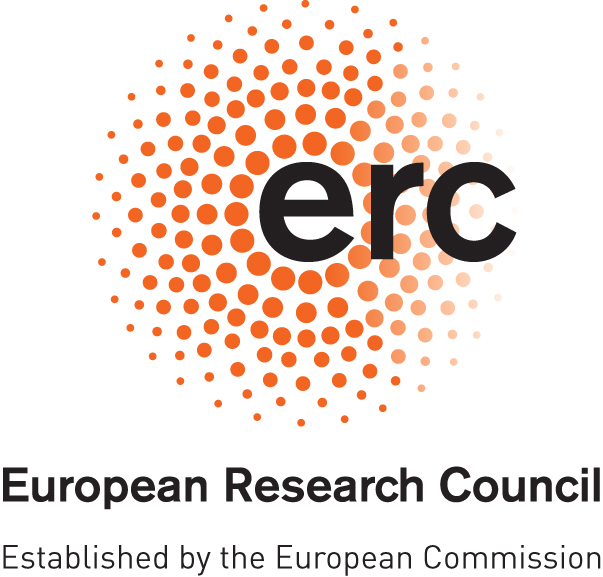}%
\end{textblock}
\begin{textblock}{20}(-0.35, 13.4)
\includegraphics[width=70px]{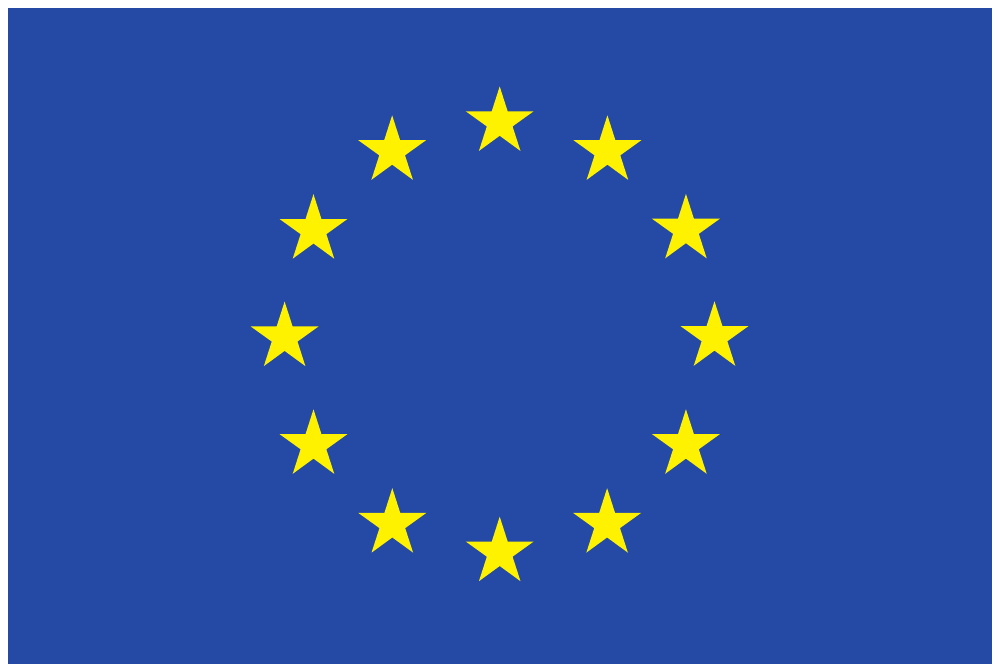}%
\end{textblock}

\begin{abstract}
    We prove a recent conjecture of Beisegel et al.~that for every positive integer $k$, every graph containing an induced $P_k$ also contains an avoidable $P_k$. 
    Avoidability generalises the notion of simpliciality best known in the context of chordal graphs.
    The conjecture was only established for $k \in \{1,2\}$ (Ohtsuki et al.~1976, and Beisegel et al.~2019, respectively).
    Our result also implies a result of Chv\'atal et al.~2002, which assumed cycle restrictions. 
    We provide a constructive and elementary proof, relying on a single trick regarding the induction hypothesis. 
    In the line of previous works, we discuss conditions for multiple avoidable paths to exist. 
\end{abstract}

\section{Introduction}\label{sec:intro}

A graph $G$ is \emph{chordal} if every induced cycle is of length three. 
A classical result of Dirac~\cite{dirac1961rigid} states that every chordal graph has a \emph{simplicial} vertex, that is, a vertex which neighbourhood is a clique.
However, not all graphs exhibit the nice structure of chordal graphs, and the statement does not extend to general graphs.

\subsection{From simplicial vertices to avoidable paths}

One way to generalise Dirac's result is through the following more flexible notion.

\begin{definition}[Avoidable vertex]\label{def:avoidable-vertex}
    A vertex $v$ in a graph $G$ is \emph{avoidable} if every induced path on three vertices with middle vertex $v$ is contained in an induced cycle in~$G$.
\end{definition}

Note that in a chordal graph, every avoidable vertex is simplicial. 
The next theorem can be inferred from~\cite{ohtsuki1976minimal,berry1998separability,aboulker2015vertex}; see also~\cite{beisegel2019avoidable} for a nice introduction.

\begin{theorem}\label{thm:avoidable-vertex}
    Every graph has an avoidable vertex.
\end{theorem}

Recently in~\cite{beisegel2019avoidable}, the authors considered a generalisation of the concept of avoidable vertices to edges, and extended \cref{thm:avoidable-vertex} to that notion.

\begin{definition}[Avoidable edge]\label{def:avoidable-edge}
    An edge $uv$ in a graph $G$ is \emph{avoidable} if every induced path on four vertices with middle edge $uv$ is contained in an induced cycle in~$G$.
\end{definition}

\begin{theorem}[Beisegel et al.~\cite{beisegel2019avoidable}]\label{th:edge}
    Every graph has an avoidable edge.
\end{theorem}

This notion naturally generalises to paths, as follows.

\begin{definition}[Extension]\label{def:extension}
    Given an induced path $P$ in a graph $G$, an \emph{extension} of $P$ is an induced path $xPy$ in $G$ for some vertices $x,y$.
\end{definition}

\begin{definition}[Failing]\label{def:failing}
    An induced path $P$ in a graph $G$ is \emph{failing} if there is no induced cycle of $G$ containing $P$.
\end{definition}

\begin{definition}[Avoidable]\label{def:avoidable-path}
    A path $P$ in a graph $G$ is \emph{avoidable} if it is induced and has no failing extension.
    Given a subgraph $G'$ of $G$, we say that $P$ is an avoidable path of $G$ in $G'$ if it is avoidable in $G$ and $V(P) \subseteq V(G')$.
\end{definition}

A graph $G$ is $P_k$-free if it does not contain a $P_k$, that is, an induced path on $k$ vertices.
In~\cite{beisegel2019avoidable} the authors conjecture that for every positive integer $k$, every graph either is $P_k$-free or contains an avoidable path on $k$ vertices.
This conjecture is motivated by the following result of Chv\'atal et al.~\cite{chvatal2002dirac}, which generalises Dirac's theorem. 
A $C_{\geq p}$-free graph is a graph where every induced cycle has at most $p-1$ vertices. 
The $C_{\geq 4}$-free graphs are exactly the chordal graphs. 
Unless specified otherwise, we consider cycles to be induced. 

\begin{theorem}[Chv\'atal et al.~\cite{chvatal2002dirac}]\label{thm:chvatal}
    For every positive integer $k$, every $C_{\geq k+3}$-free graph
    either is $P_k$-free or contains an avoidable path on $k$ vertices.
\end{theorem}

In fact, \cref{thm:chvatal} originally states the existence of a simplicial path in the class of $C_{\geq k+3}$-free graphs.
A \emph{simplicial} path is an induced path with no extension: it is avoidable by vacuity. 
Note that these two definitions coincide in such a class, as no cycle on at most $k+2$ vertices can contain the extension of an induced path on $k$ vertices.

Here, we confirm the aforementioned conjecture~\cite[Conjecture 1]{beisegel2019avoidable}, as follows.

\begin{theorem}\label{thm:existsavoidablePk}%
    For every positive integer $k$, every graph either is $P_k$-free or contains an avoidable $P_k$.
\end{theorem}

In fact, we prove \cref{thm:existsavoidablePk} using a stronger induction hypothesis, in the exact same flavour as~\cite{chvatal2002dirac}, see \cref{th:doubleinduction} in \cref{sec:proof}.

\subsection{Consequences}\label{subsec:cons}

We point out that the proof of \cref{thm:existsavoidablePk} is self-sufficient, thus this supersedes the arguments for \cref{thm:avoidable-vertex,th:edge,thm:chvatal}.

By using ingredients of \cref{th:doubleinduction} (namely \cref{lem:mergingvertices}), we obtain a way to build more than one avoidable $P_k$.

\begin{corollary}\label{cor:connectedsubset}
    For every positive integer $k$, graph $G$ and subset $X \subseteq V(G)$ such that $G[X]$ is connected, either $G - N[X]$ is $P_k$-free or there is an avoidable $P_k$ of $G$ in $G-N[X]$.
\end{corollary}

\begin{corollary}\label{cor:2indPk}
    For every positive integer $k$ and graph $G$, either $G$ does not contain two non-adjacent $P_k$, or it contains two non-adjacent avoidable $P_k$.
\end{corollary}

Since \cref{cor:2indPk} is not as straightforward as its predecessor, we include a proof.

\begin{proof}
    Let $Q_1$ and $Q_2$ be two non-adjacent $P_k$.
    By \cref{cor:connectedsubset}, either $G - N[Q_1]$ is $P_k$-free or there is an avoidable $P_k$ of $G$ in $G - N[Q_1]$.
    The first outcome is ruled out by the existence of $Q_2$.
    Let $Q'_2$ be an avoidable $P_k$ of $G$ in $G - N[Q_1]$.
    We repeat the argument with $Q'_2$ instead of $Q_1$, and obtain an avoidable $P_k$ of $G$ in $G - N[Q'_2]$, call it~$Q'_1$.
    The two paths $Q'_1$ and $Q'_2$ are two non-adjacent avoidable $P_k$, as desired.
\end{proof}

We can also wonder:

\begin{question}\label{qu:2disjointPk}
For every positive integer $k$, does every graph $G$ either not contain two disjoint $P_k$, or contain two disjoint avoidable $P_k$?
\end{question}

\def\k{2}
We know the answer to be positive in the case $k \in \{1,2\}$, due to~\cite[Theorems 3.3 and 6.4]{beisegel2019avoidable}.
The answer turns out to be negative in all other cases, as exhibited in the following counter-example for $k\geq 3$, which consists of a cycle on $2k-1$ vertices with an added vertex adjacent to two consecutive vertices on the cycle (see \cref{fig:counter-example-disjoint} for the case $k = 3$).
This graph contains two disjoint $P_k$, and it has $2k$ vertices, so any two disjoint $P_k$ are in fact complementary in the graph.
Suppose that it contains two disjoint avoidable $P_k$, and note that each intersects the triangle (otherwise the complement would not be a path).
Since there are three vertices in the triangle, there is an avoidable $P_k$ containing a single vertex in the triangle.
This $P_k$ has a failing extension, a contradiction.

\begin{figure}[ht]
    \centering

\begin{tikzpicture}[scale = 0.9]

\def\nodesc{0.9}
\def\r{1.9}
\def\a{180}
\def\lw{1.8pt}
\def\lwb{.9pt}
\def\lwc{4.9pt}

\pgfmathsetmacro\Ntemp{2*\k};
\pgfmathtruncatemacro\N{round(\Ntemp)}
\pgfmathsetmacro\Km{\k-1};
\pgfmathtruncatemacro\kmi{round(\Km)};
\pgfmathsetmacro\Ntot{\N+1};
\pgfmathsetmacro\kpi{\k+1};
\pgfmathtruncatemacro\kp{round(\kpi)};

\foreach \i in {0,...,\N}
{
	\ifthenelse{\i < \kp}
	{\def\c{blue!15}}
	{\def\c{red!15}}
	\node[scale = \nodesc, nodes={\c}{}{}{},line width = \lwb] (u\i) at (\i*360/\Ntot+\a:\r) {};
}

\node[scale = \nodesc, nodes={red!15}{}{}{},line width = \lwb] (uf) at (0:0) {};

\foreach \i in {0,...,\kmi}
{
	
	\pgfmathsetmacro\ipi{1+\i};
	\pgfmathtruncatemacro\ip{round(\ipi)};
	\draw[line width = \lw, blue!60] (u\i) to (u\ip);
	
	\pgfmathsetmacro\ibistemp{\i+\k+1};
	\pgfmathtruncatemacro\ibis{round(\ibistemp)};
	\pgfmathsetmacro\ibisptemp{1+\ibis};
	\pgfmathtruncatemacro\ibisp{round(\ibisptemp)};
	
	\ifthenelse{\i < \kmi}
	{\draw[line width = \lw, red!60] (u\ibis) to (u\ibisp);}
	
}

\draw[line width = \lwb] (u0) to (u\N);
\draw[line width = \lwc, green!25]  (u\k) to (u\kp);
\draw[line width = \lwb] (u\k) to (u\kp);
\draw[line width = \lwc, green!25] (u0) to (uf);
\draw[line width = \lwb] (u0) to (uf);
\draw[red!60, line width=\lw](u\N) to (uf);

\end{tikzpicture}
    \caption{A graph that contains two disjoint $P_3$ (in blue and in red) but no two disjoint avoidable $P_3$ (there is a unique partition into two disjoint $P_3$, up to symmetry).
    In green, a failing extension of the blue path.}
    \label{fig:counter-example-disjoint}
\end{figure}
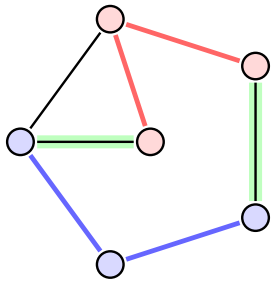

In \cref{sec:algorithm}, we present a concise algorithm which follows the proof of \cref{th:doubleinduction}.
As discussed there, the algorithm has complexity $O(n^{k+2})$ which, while naive, is the right order of magnitude under ETH. 

\section{A stronger induction hypothesis}\label{sec:proof}

All graphs considered in this paper are finite, simple and loopless.
Given a graph $G$, we denote by $V(G)$ its set of vertices, and by $E(G)\subseteq \{\{x,y\} \mid x,y\in V(G),\ x\neq y\}$ its set of edges.  
Edges are denoted by $uv$ (or $vu$) instead of $\{u,v\}$.
If $uv$ is an edge, then we say that $u$ and $v$ are \emph{adjacent}.
Given a vertex $u$, the neighbourhood $N(u)$ of $u$ is the set of vertices of $G$ that are adjacent to $u$. 
The closed neighbourhood $N[u]$ of $u$ is the set $N(u)\cup \{u\}$. 
If $X\subseteq V(G)$, then we define $N[X] \coloneqq \bigcup_{x\in X} N[x]$ and $N(X) \coloneqq N[X]\setminus X$. 
The subgraph of $G$ \emph{induced} by $X$, denoted by $G[X]$, is the graph $(X,E(G)\cap \{\{x,y\} \mid x,y\in X,\ x\neq y\})$, and $G-X$ is the graph $G[V(G)\setminus X]$.
Given two adjacent vertices $u_{1}$ and $u_{2}$ of $G$, the graph obtained by \emph{merging} $u_1$ and $u_{2}$ is the graph obtained from $G$ by replacing $u_{1}$ and $u_{2}$ with a new vertex $u$ such that $N(u)=N(\{u_1,u_2\})$.
Given a graph $G$ and two subsets $X$ and $Y$ of $V(G)$, we say that $X$ \emph{dominates} $Y$ if every vertex of $Y \setminus X$ has a neighbour in $X$ (equivalently, if $Y \subseteq N[X]$).
\medskip

We first define two useful properties.
\begin{definition}[Basic property $H_B$]\label{def:HB}
    Given a positive integer $k$ and a graph $G$, the property $H_B(G,k)$ holds if either $G$ is $P_k$-free or there is an avoidable $P_k$ in $G$.
\end{definition}

\begin{definition}[Refined property $H_R$]\label{def:HR}
    Given a positive integer $k$, a graph $G$ and a vertex $u \in V(G)$, the property $H_R(G,k,u)$ holds if either $G-N[u]$ is $P_k$-free or there is an avoidable $P_k$ of $G$ in $G-N[u]$.
    
    Given a positive integer $k$ and a graph $G$, the property $H_R(G,k)$ holds if $H_R(G,k,u)$ holds for every $u \in V(G)$.
\end{definition}

Note that property $H_R$ does not directly imply property $H_B$. 
We also emphasise the fact that an avoidable path in a subgraph is not necessarily an avoidable path in the whole graph.

We now prove a form of heredity in $H_R$.

\begin{lemma}\label{lem:mergingvertices}
    Let $k$ be a positive integer, $G$ a graph and $u_1u_2$ an edge of $G$.
    Let $G'$ be the graph obtained from $G$ by merging the two vertices $u_1$ and $u_2$ into one vertex $u$. 
    If $G' - N[u]$ contains a $P_k$, then $H_R(G',k,u)$ implies $H_R(G,k,u_1)$. 
\end{lemma}

\begin{proof}
    Suppose $G' - N[u]$ contains a $P_k$, and that $H_R(G',k,u)$ holds but not $H_R(G,k,u_1)$.
    Since $G'-N[u]$ is not $P_k$-free, there is an avoidable $P_k$ of $G'$ in $G'-N[u]$.
    Call it $Q$.
    The path $Q$ is contained in $G'- N[u]=G - N[\{u_1,u_2\}]$, so in particular in $G -  N[u_1]$.
    Since $H_R(G,k,u_1)$ does not hold, $Q$ is not an avoidable $P_k$ of $G$.
    Thus, there is a failing extension $xQy$ of $Q$ in $G$.
    Note that $x, y, u_1$, and $u_2$ are all pairwise distinct.

    Hence, $xQy$ is an extension of $Q$ in $G'$, and there is an induced cycle $C$ in $G'$ containing the path $xQy$.
    If $u \not\in C$, then the cycle $C$ is also an induced cycle in $G$ containing $xQy$, a contradiction.
    Therefore, $u \in C$.
    By replacing $u$ with either $u_1$, $u_2$ or the edge $u_1u_2$ as appropriate, we obtain an induced cycle in $G$ containing $xQy$, a contradiction.
\end{proof}

We are now ready to prove the main technical result of this paper.

\begin{theorem}\label{th:doubleinduction}
    For every positive integer $k$ and every graph $G$, both properties $H_B(G,k)$ and $H_R(G,k)$ hold.
\end{theorem}

\begin{proof}
    Suppose the statement is false and consider a counter-example $G$ which is minimal with respect to the number of vertices.
    \begin{lemma}\label{lem:HR}
        The property $H_R(G,k)$ holds for every $k$.
    \end{lemma}
    \begin{proof}
        We proceed by contradiction.
        Suppose that $H_R(G,k,u)$ does not hold for some $k$ and some vertex $u \in V(G)$, that is, there exists a $P_k$ in $G-N[u]$, and every $P_k$ in $G-N[u]$ has a failing extension.
        We prove the following.
        
        \needspace{0.3in}
        \begin{claim}\label{cl:PkdominatesNu}
            Every $P_k$ in $G - N[u]$ dominates $N(u)$.
        \end{claim}
        
        \begin{proof}
            Assume towards a contradiction that there is a $P_k$ in $G - N[u]$, call it $Q$, which is not adjacent to some vertex $v \in N(u)$.
            Then $G  - N[\{u,v\}]$ contains a $P_k$. 
            Let $G'$ be the graph obtained from $G$ by merging $u$ and $v$ into a vertex $u'$.
            Since $G'$ has fewer vertices than $G$, the property $H_R(G',k,u')$ holds by minimality of $G$.
            By \cref{lem:mergingvertices}, the property $H_R(G,k,u)$ holds, a contradiction.
            \cqed
        \end{proof}
        
        Let $G' \coloneqq G-N[u]$. 
        Then $G'$ contains a $P_k$. 
        As $G'$ contains fewer vertices than~$G$, the property $H_B(G',k)$ holds. 
        Let $Q$ be an avoidable $P_k$ of $G'$. 
        By assumption, $Q$ is not an avoidable $P_k$ of $G$. 
        So there is a failing extension $xQy$ of $Q$ in $G$.
        Since $Q$ has no failing extension in $G'$, we can assume without loss of generality that $y \in N(u)$. 
        It follows that $x \not\in N(u)$: otherwise the cycle $x Q y u$ contradicts the fact that $x Q y$ is failing. 
        By definition of an extension, $x Q y$ is an induced path.
        Let $z$ be the only neighbour of $y$ in~$Q$, and let us now consider the path $x Q - z$. 
        It is a $P_k$, and it does not intersect $N[u]$.
        However, no vertex in it is adjacent to $y$ which lies in $N(u)$, contradicting \cref{cl:PkdominatesNu}.
        \lqed
    \end{proof}
    
    \begin{lemma}\label{lem:HB}
        The property $H_B(G,k)$ holds for every $k$.
    \end{lemma}
    
    \begin{proof}
        Assume towards a contradiction that for some $k$, the property $H_B(G,k)$ does not hold.
        By \cref{lem:HR}, the property $H_R(G,k,u)$ holds for every vertex $u \in V(G)$.
        In other words, the graph $G$ contains a $P_k$ but no avoidable $P_k$, and for every vertex $u \in V(G)$, either $G - N[u]$ is $P_k$-free or there is an avoidable $P_k$ of $G$ in $G - N[u]$. 
    
        We derive the following claim.
        \begin{claim}\label{cl:PkdominatesG}
            Every $P_k$ in $G$ dominates $V(G)$.
        \end{claim}
        \begin{proof}
            Suppose there is a $P_k$, call it $Q$, that does not dominate some vertex $u$ of $G$.
            Since $H_R(G,k)$ holds, either $G - N[u]$ is $P_k$-free or there is an avoidable $P_k$ of $G$ in $G - N[u]$.
            The first case contradicts the existence of $Q$, and the second contradicts the fact that $H_B(G,k)$ does not hold.
        \cqed
        \end{proof}
    
        Since $H_B(G,k)$ does not hold, $G$ contains a $P_k$, say $Q$, that is not avoidable.
        So it has a failing extension $xQy$.
        Let $z$ be the only neighbour of $y$ in $Q$, and consider the path $xQ-z$.
        It is an induced $P_k$ and none of its vertices is adjacent to $y$.
        This contradicts \cref{cl:PkdominatesG}.
        \lqed
    \end{proof}

    Finally, \cref{lem:HR,lem:HB} together contradict $G$ being a counter-example.
\end{proof}

\Cref{thm:existsavoidablePk} directly follows from \cref{th:doubleinduction}.

\section{An algorithm for \texorpdfstring{\cref{th:doubleinduction}}{Theorem 2.4}}\label{sec:algorithm}

By going through the proof and extracting the key ingredients, we obtain a straightforward algorithm verifying both properties (see \cref{alg:findPk}).

\begin{algorithm}[ht]
    \caption{finds an avoidable path of given length in a given graph, if any.}\label{alg:findPk}
    \begin{algorithmic}[1]
        \Procedure{FindAvoidablePathRefined}{$G,k,u$}
            \ForAll{$v \in N(u)$}
                \If{\Call{InducedPath}{$G-N[\{u,v\}],k$}$\,\neq \texttt{null}$}
                    \State $G' \gets G$ with $u$ and $v$ merged into $u'$
                    \State \Return \Call{FindAvoidablePathRefined}{$G',k,u'$}
                \EndIf
            \EndFor\vspace{-0.1cm}
            \State \Return \Call{FindAvoidablePath}{$G-N[u],k$}
        \EndProcedure
        \vspace{0.1cm}
        \Procedure{FindAvoidablePath}{$G,k$}
            \ForAll{$u \in V(G)$}
                \If{\Call{InducedPath}{$G-N[u],k$}$\,\neq \texttt{null}$}
                    \State \Return \Call{FindAvoidablePathRefined}{$G,k,u$}
                \EndIf
            \EndFor\vspace{-0.1cm}
            \State{\Return \Call{InducedPath}{$G,k$}}
        \EndProcedure
    \end{algorithmic}
\end{algorithm}

The algorithm uses the subprocedure \textsc{InducedPath} that, given a graph $G$ and a positive integer $k$, decides whether $G$ contains a $P_k$. 
If it does, the procedure returns one, otherwise it returns $\texttt{null}$.
The naive algorithm for that (testing all subsets of size $k$) has complexity $O(n^k)$. 
However, this is nearly optimal. 
Indeed, the problem of finding a $P_k$ in a given graph is W[1]-hard\footnote[1]{see e.g.~\cite{cygan2015parameterized} for definitions around complexity} when parametrised by $k$ (see~\cite[Ex.\@ 13.16, p.\@ 460]{cygan2015parameterized}). 
In fact, the hinted reduction has a linear blow-up, so it follows that there is no $f(k)\cdot n^{o(k)}$ algorithm under ETH.

Let $k$ be a positive integer, and let $B(n)$ (resp.~$R(n)$) be the worst case complexity of \textsc{FindAvoidablePath} (resp.~\textsc{FindAvoidablePathRefined}) on an $n$-vertex graph with parameter $k$.
We have $B(n) \leq n \cdot n^k + \max(R(n),n^k)$, and $R(n) \leq n \cdot n^k + \max(R(n-1),B(n-2))$.
We obtain $R(n) \leq n^{k+2}$ and $B(n) \leq n^{k+2}+n^{k+1}$.
While this may well be improved, the known limitations for finding an induced path on $k$ vertices also apply for an induced avoidable path on $k$ vertices (by \cref{thm:existsavoidablePk}, if the first exists, then so does the second).
Therefore, the order of magnitude of this naive algorithm is correct. 

Note that there is a yet more naive algorithm blindly checking for every subset of size $k$ if it corresponds to an avoidable path.
That algorithm has comparable complexity to ours (though slightly worse, at least at first sight).
However, we wanted to emphasise that our proof of \cref{th:doubleinduction} is constructive and yields an elementary algorithm.
Also, we believe that it provides an outline of the proof which might be helpful to the reader.

\section{Conclusion}\label{sec:ccl}

Given the discussions in \cref{subsec:cons}, it is tempting to ask when a graph admits three (or more) disjoint (resp.\@ pairwise non-adjacent) avoidable paths.
Note that though \cref{cor:connectedsubset} arms us with sufficient conditions for there to be more than two avoidable $P_k$, we do not believe that the corresponding sufficient conditions are necessary.
However, it seems the picture is murky already for chordal graphs.

It is tempting to wonder whether we can obtain another avoidable structure.
Though in some cases the very notion of extension becomes unclear (what should an extension of a clique be?), it does not seem like any other structure survives the test of chordal graphs or simple ad hoc constructions\textemdash even when allowing a family of graphs instead of fixing a single pattern (like a path on $k$ vertices).
This motivates us to formulate the following question.

\begin{question}\label{qu:otherstructure}
    Does there exist a family $\mathcal{H}$ of connected graphs, not containing any path, such that any graph is either $\mathcal{H}$-free or contains an avoidable element of $\mathcal{H}$?
\end{question}

The notion of avoidability in this context is deliberately left up to interpretation.

\section*{Acknowledgements}

We gratefully acknowledge support from Nicolas Bonichon and the Simon family for the organisation of the $4^{\textrm{th}}$ Pessac Graph Workshop, where this research was done. We are indebted to Micha\l \ Pilipczuk for providing helpful references regarding the complexity of finding an induced path of given length. Last but not least, we thank Peppie for her unwavering support during the work sessions.

\bibliography{pessacavoidable}

\end{document}